\documentclass{amsart}
\usepackage[english]{babel}
\usepackage[utf8]{inputenc}
\usepackage{bbm}
\usepackage{verbatim}
\usepackage{graphicx}
\usepackage{subfig}
\usepackage{color}

\newcommand{\N}{\ensuremath{\mathbb{N}}}
\newcommand{\R}{\ensuremath{\mathbb{R}}}

\newcommand{\Z}{\ensuremath{\mathbb{Z}}}
\newcommand{\E}{\ensuremath{\mathbb{E}}}
\renewcommand{\P}{\ensuremath{\mathbb{P}}}

\newcommand{\eps}{\ensuremath{\varepsilon}}
\newcommand{\ind}[1]{\ensuremath{\mathbbm{1}_{\left\{#1\right\}}}}
\newcommand{\cal}[1]{\ensuremath{\mathcal{#1}}}

\newcommand\croc[1]{\left\langle #1\right\rangle}

\def\var{\mathrm{var}}
\newcommand{\diff}{\mathop{}\mathopen{}\mathrm{d}}

\newtheorem{definition}{Definition}
\newtheorem{proposition}{Proposition}
\newtheorem{corollary}{Corollary}[proposition]

\newtheorem{lemma}{Lemma}
\newtheorem{theorem}{Theorem}
\graphicspath{{Fig/}}

\title[Autoregulation of Gene Expression]{A Stochastic Analysis of Autoregulation of Gene Expression}

\date{\today}

\author[R. Dessalles]{Renaud Dessalles}
\email{Renaud.Dessalles@inria.fr}
\author[V. Fromion]{Vincent Fromion}
\email{Vincent.Fromion@jouy.inra.fr}
\address[R. Dessalles, V. Fromion]{MaIAGE, INRA,
Domaine de Vilvert, 78350 Jouy-en-Josas, France}
\author[Ph. Robert]{Philippe Robert}
\email{Philippe.Robert@inria.fr}
\urladdr{http://www-rocq.inria.fr/\string~robert}

\address[R. Dessalles,Ph. Robert]{INRIA Paris---Rocquencourt, Domaine de Voluceau, 78153 Le Chesnay, France}


\date{\today}

\begin{document}

\begin{abstract}
This paper  analyzes,  in the context of  a prokaryotic cell,  the stochastic variability of  the number  of proteins when there is a control of gene expression by an autoregulation scheme.  The goal of this work is to estimate the efficiency of the regulation  to limit the  fluctuations of the number of copies of a given protein. The autoregulation considered in this paper relies mainly on a negative feedback: the proteins are repressors of their own gene expression.  The efficiency of a production process without feedback control is compared to a production process with an autoregulation of the gene expression assuming that both of them produce the same  average number of proteins.  The main characteristic used for the comparison is the standard deviation of the number of proteins at equilibrium.  With a  Markovian representation  and a simple model of repression, we prove that, under a scaling regime, the repression mechanism follows a Hill repression scheme with an hyperbolic control.  An explicit asymptotic expression of the variance of the number of proteins under this regulation mechanism is obtained. Simulations are used to study other aspects of autoregulation such as the rate of convergence to equilibrium of the production process and the case where the control of the production process of proteins is achieved via the inhibition of mRNAs. 
\end{abstract}

\maketitle 

\bigskip

\hrule

\vspace{-3mm}

\tableofcontents

\vspace{-1cm}

\hrule

\bigskip

\section{Introduction}
\subsection{Biological Context}
The \emph{gene expression} is the process by which  genetic information is  used to produce  functional products of gene expression: proteins and non-coding RNAs. This paper concerns itself with the production of proteins. The information flow from DNA genes to proteins is a fundamental process. It is composed of three main steps: \emph{Gene Activation}, \emph{transcription} and \emph{translation}. 

\begin{enumerate}
\item  The initiation of transcription is strongly regulated.  Schematically the gene   is said to be in ``inactive state'' if a repressor is bound on the gene's promoter preventing the RNA polymerase from binding and is in ``active state'' otherwise. 
\item When the gene is in active state, the RNA polymerase binds and initiates transcription that leads to the creation of a mRNA, a copy of a specific DNA sequence.
\item The translation of the messenger into a protein is achieved by a large complex molecule: the \emph{ribosome}. A ribosome binds to an active mRNA, initiates the translation and proceeds to protein elongation. Once the elongation terminates, the protein is released in the medium and the ribosome is anew available for any another translation.
\end{enumerate}

The production of proteins is the most important cellular activity, both for the functional role and the high associated cost in terms of resources. In a \emph{E. Coli} bacterium for example there are about $3.6\times 10^6$ proteins of approximately $2000$ different types with a large variability in concentration, depending on their types: from a few dozen up to $10^5$.  The gene expression is additionally a highly stochastic process and results from the realization of a very large number of elementary stochastic processes of different nature. The three main steps are the results of a large number of encounters of macromolecules following random motions, due in particular to thermal excitation, in the viscous fluid of the cytoplasm. One of the key problems is to understand the basic mechanisms which allow a cell to produce a large number of proteins with very different concentrations and in a random context. This can be seen as a problem of minimization of the variance of the number of proteins of each type. 

To study this problem, one can take a simple stochastic model, with a limited set $S$ of parameters preferably, describing the three steps of the production of a given type of protein. Once a closed form expression of the  variance of  the number of proteins is obtained, it is natural to find the parameters of the set $S$ which minimizes the variance with the constraint that the mean number of proteins is fixed. See the survey Paulsson~\cite{Paulsson}. 

A more effective way to regulate the number of proteins can be of using a direct feedback control, an {\em autoregulation} mechanism, so that the production of proteins is either sped up or slowed down depending on the current number of proteins. It should be noted that the feedback control loop can involve other intermediate proteins to achieve this goal, like the classical lac operon,  but it is not considered here. See Yildirim and MacKey~\cite{Yildirim} for example. 

The protein can regulate the gene activation simply, for example by being a repressor and tend to bind on his own gene's promoter. This is the {\em autogenous regulation} scheme. See Goldberger~\cite{Goldberger} and  Maloy and Stewart~\cite{Maloy}. See also Thattai and van Oudenaarden~\cite{Thattai}. Other autoregulation mechanisms are possible in cells, such as an autoregulation on the mRNAs where a protein inhibits its own translation initiation by binding to the translation initiation region of  its own mRNAs. It occurs for example in the production of ribosomal proteins, see Kaczanowska and Ryd\'en-Aulin~\cite{Kaczanowska}.   The idea being that a feedback mechanism may reduce significantly the number of large excursions from the mean.  In this paper, the mathematical analysis will  mainly focus  on a negative autogenous feedback, when the rate of inactivation of the  gene expression grows with the number of proteins. 

\subsection{Literature}
The classical results concerning the mathematical analysis of the variance of the number of proteins has been investigated in Berg~\cite{Berg1978} and Rigney~\cite{Rigney1979,Rigney1977} and reviewed more recently by Paulsson~\cite{Paulsson}, see also Raj and van Oudenaarden~\cite{Raj} for the biological aspects. These references use the three stage model, the state of the system is given by three variables: the state of the promoter, the number of mRNAs and the number of proteins. Mathematically, the techniques used rely on the Fokker-Planck equations of the associated three dimensional Markov process and the observation that at equilibrium, a recurrence on the moments of the number of proteins holds. Fromion et al.~\cite{Fromion} investigates  a more general model (elongation times are not necessarily exponentially distributed in particular)  and an alternative technique to a Markovian approach is introduced.

Concerning the evaluation of autoregulation, most of mathematical models use a continuous state, the rate of production of proteins depends linearly on the number of mRNAs and the rate of production of mRNAs is a function $k(p)$ exhibiting a non-linear dependence on the current number $p$ of proteins. In Rosenfeld et al.~\cite{Rosenfeld} and Becskei and Serrano~\cite{Becskei}, based on experiments the constant $k(p)$ is taken a {\em Hill repression function}, i.e. $k(p)=a/(b+p^n)$ for some constants $a$ and $b$ and $n\geq 1$ is the Hill coefficient. See also Thattai and van Oudenaarden~\cite{Thattai}. Related models in a similar framework with further results are presented in Bokes et al.~\cite{Bokes} and Yvinec et al.~\cite{Yvinec}. For most of these models the state of the promoter, active or inactive, which is a source of variability is not taken into account, it is in some way encapsulated in the constant $k(p)$ whose representation is rarely discussed.  In Hornos et al.~\cite{Hornos} the state of the gene expression, on or off, is taken into account but not the number of mRNAs and therefore the fluctuations generated by transcription.  The parameter of activation $k(p)$ is of course crucial in our case since autogenous regulation rely on the state of the promoter which can be inactivated by proteins. Our model includes it.  See also Fournier et al.~\cite{Fournier} for some simulations of these stochastic models of autoregulation as well as some experiments.

\subsection{Results of the Paper}
The main goal of this paper is to estimate the possible benefit of the autogenous regulation  to control the fluctuations of the number of copies of a given protein.  The efficiency of a production process without feedback control is compared to a production process with an autoregulation of the gene expression, assuming that both of them produce the same average number  proteins.  The main characteristic used for the comparison is the standard deviation of the number of proteins at equilibrium. For this purpose, two approaches are used. 

\medskip
\noindent
{\sc Mathematical Analysis.} One first studies the distribution of the number of proteins via a stochastic model. When there is no regulation, the corresponding classical mathematical model has been investigated in detail for some time now. In particular, the standard deviation of the number of proteins at equilibrium has a closed form expression in terms of the basic parameters of the production process. See for example the survey Paulsson~\cite{Paulsson}, and also Fromion et al.~\cite{Fromion}.

To represent the negative feedback of the autogenous regulation, a simple model is used: each protein can be bound, at some rate and for some random duration of time, on its own gene expression. In this situation the gene expression is inactive and the transcription is not possible during that time. This amounts to say that the gene expression is deactivated at a rate proportional to the number of proteins. The activation rate is constant. 

 As  will be seen, the mathematical model of the autogenous regulation is more complicated, in particular there is no recurrence relationship between the moments of the number of proteins at equilibrium as in the classical model of protein production process.  For this reason, a limiting procedure is used, it amounts to assume that the dynamics of the activation of the gene expression and of the evolution of mRNAs occur on a much faster time scale than the dynamics of the proteins. The values of the key parameters are presented in Section~\ref{biopar}.  The scaling parameter is the multiplicative factor describing the difference of speed  of these two time scales. The main convergence result is  Theorem~\ref{TheoFeed}.  The assumption of a fast time scale for gene expression activation and mRNAs  is quite common in the literature, see Bokes et al.~\cite{Bokes} and Yvinec et al.~\cite{Yvinec}. The techniques used in these references rely on singular perturbation methods to deal with the two time scales. In our setting,  a probabilistic approach is used, as  will be seen, it gives precise results on the asymptotic stochastic evolution of the number of proteins.  

Under this limiting regime it is shown that,  asymptotically,  the protein production process can be described as a birth and death process. See Keilson~\cite{Keilson} for example. In state $x\in\N$, the birth rate is given $a/(b+x)$ for some constants $a$ and $b$. This is a contribution of the paper that, with a simple model of the autoregulation, one can show that the repression mechanism follows indeed a Hill repression scheme with an hyperbolic control, i.e. with Hill coefficient $1$.  The death rate is not changed by the limiting procedure, it is  proportional to $x$.  Consequently, one can get an asymptotic closed form expression of the standard deviation of the number of proteins by using the explicit  representation of the equilibrium of this birth and death process. See Corollary~\ref{corol1}. It is shown that, in this limiting regime, the standard deviation is reduced by 30\%. The corresponding results are presented in Section~\ref{scalingsec} and Section~\ref{SecInv} and in Appendix~\ref{Appendix}. The mathematical results are obtained via convergence theorems for sequence of Markov process, the proof of a stochastic averaging principle and a saddle point approximation result. 

\medskip
\noindent
{\sc Simulations.}
We also analyze, via simulations,  autogenous regulation but also other aspects related to the regulation of protein production. This is presented in Section~\ref{SecDisc}. Simulations are used mainly because
of the complexity of the mathematical models  of some aspects of the autogenous regulation. By using plausible biological parameters, one gets an improvement of 15\% for the standard deviation of the number of proteins can be expected.  This is significantly less than the performances of the limiting mathematical model studied in Section~\ref{scalingsec}. The main reason seems to be that that the scaling parameter is not, in some cases, sufficiently large to have a reasonable accuracy with the limit given by the convergence result of Theorem~\ref{TheoFeed}.

Via simulations, one also investigates the case when the regulation is not on the gene expression but on the corresponding mRNAs: a protein can block an mRNA for some time. In this situation, it could be expected that the production process is modulated more smoothly by playing on the inactivation of a fraction of the mRNAs and not on the rough on-off control of the gene expression. It is shown that the improvement is real but not that big (less than 10\%). It is nevertheless remarkable that if the average life time of mRNAs is significantly increased, our experiments show  that the benefit of such regulation can be of the order of more than 30\% on the standard deviation of the number of proteins. 

Coming back to regulation on the gene expression. Our experiments show that, despite the  impact of autogenous regulation  on  fluctuations of the number of proteins can be limited,  it has nevertheless a very interesting property. Starting with a number of proteins significantly less (or greater) than the average number of proteins at equilibrium, the autogenous regulation returns to the ``correct'' number of proteins much faster than the classical production process without regulation. This is a clear advantage of this mechanisms to adapt quickly when  biological conditions change due to an external stress for example.  See Section~\ref{Vers}.  This phenomenon has been observed, via experiments, in Rosenfeld et al.~\cite{Rosenfeld}. See also Camas et al.~\cite{Camas}. Finally Section~\ref{FreqSec} investigates the comparison of production processes with and without a feedback on the gene expression  through the estimation of their respective power spectral density. 

\medskip
\noindent
{\bf Acknowledgments}\\
The authors are very grateful to an anonymous referee for the important work done on an earlier version of the paper. 
\section{Stochastic Models of Protein Production}\label{secMod}
We present the stochastic models used to investigate the protein production process. We will use the three step model describing the activation-deactivation of the gene, the transcription phase and the translation phase. 
 Like in most of the literature, it is assumed that the various events, like the encounter of two macromolecules,  occurring within the cell have a duration with an exponential distribution. We start with the classical model used in this domain since the late 70's by Berg~\cite{Berg1978} and Rigney~\cite{Rigney1979,Rigney1977}. See also Thattai and van Oudenaarden~\cite{Thattai} and Paulsson~\cite{Paulsson}. 

\subsection{The Classical Model of Protein Production}

\begin{enumerate}
\item The gene is activated at rate $\lambda_1^+$ and inactivated at rate $\lambda_1^-$.
\item If the gene is active, an mRNA is produced at rate $\lambda_2$.  An mRNA is degraded at rate $\mu_2$. 
\item Given $M$ mRNAs at some moment, a protein is produced at rate $\lambda_3 M$. Each protein is degraded at rate $\mu_3$. 
\end{enumerate}

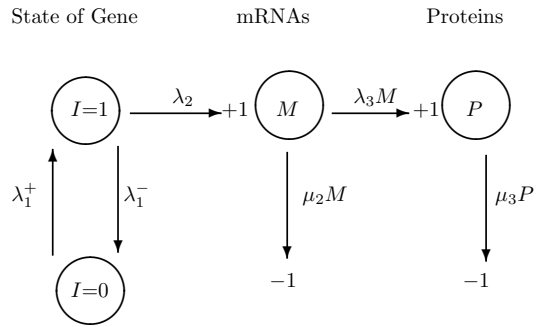
\begin{figure}
\setlength{\unitlength}{2044sp}

\begin{center}
\scalebox{0.8}{
\begin{picture}(7017,5292)(1828,-4594)
{\thicklines
\put(2395,-1328){\circle{1026}}
}%
{\thicklines\put(2465,-4058){\circle{1026}}
}%
\put(2165,-4158){$I{=}0$}
{\thicklines\put(5479,-1238){\circle{1026}}
}%
\put(4679,0){mRNAs}
\put(5279,-1408){$M$}
\put(1265,0){State of Gene}
\put(2165,-1408){$I{=}1$}
\put(3700,-1208){$\lambda_2 $}
\put(6450,-1208){$\lambda_3 M$}
\put(1265,-2700){$\lambda_1^+$}
\put(2965,-2700){$\lambda_1^- $}
\put(5679,-2700){$\mu_2 M$}
\put(8109,-4000){$-1$}
\put(8600,-2700){$\mu_3 P$}
{\thicklines\put(8309,-1238){\circle{1026}}
}%
\put(8159,-1408){$P$}
\put(7559,0){Proteins}
{\thicklines\put(3061,-1366){\vector( 1, 0){1390}}
}%
\put(4451,-1400){$+1$}
{\thicklines\put(6121,-1366){\vector( 1, 0){1175}}
}%
\put(7351,-1400){$+1$}
{\thicklines\put(1891,-3526){\vector( 0, 1){1620}}
}%
{\thicklines\put(2881,-1861){\vector( 0,-1){1620}}
}%
{\thicklines\put(5446,-1951){\vector( 0,-1){1620}}
}%
\put(5179,-4000){$-1$}
{\thicklines\put(8461,-1996){\vector( 0,-1){1620}}
}%
\end{picture}
}
\end{center}
\caption{Classical Three Stage Model for Protein Production.}\label{transfig1}
\end{figure}

The stochastic processes describing the protein production process are: $I(t)$ the state of the gene at time $t$ which is $0$ if it is inactive and $1$ otherwise. The number of mRNA at time $t$ is $M(t)$ and $P(t)$ denotes the number of proteins at that moment.  The process $(I(t),M(t),P(t))$ is Markovian with state space
\[
{\cal S}\stackrel{\text{def.}}{=}\{0,1\}\times\N^2,
\]
its transition rates are given by, if $(I(t),M(t),P(t)){=}(i,m,p)\in{\cal S}$,
\[
\begin{cases}
(0,m,p)\rightarrow (1,m,p) \quad\text{ at rate } \lambda_1^+,\qquad& (1,m,p)\rightarrow (0,m,p)\hfill \text{ at rate } \lambda_1^-,\\
(i,m,p)\rightarrow (i,m+1,p) \text{\phantom{at rate}} \lambda_2 i,\qquad &(i,m,p)\rightarrow (i,m-1,p) \hfill \mu_2m,\\
(i,m,p)\rightarrow (i,m,p+1) \text{\phantom{at rate}} \lambda_3m,\qquad& (i,m,p)\rightarrow (i,m,p-1) \hfill \mu_3p.
\end{cases}
\]
See Figure~\ref{transfig1}. 
This Markov process has  a unique invariant distribution. An explicit expression of the distribution of $P$ at equilibrium is not known but, due to the linear transition rates, the moments of $P$ can be calculated recursively. In the following $(I,M,P)$ will denote  random variables whose law is invariant for $(I(t),M(t),P(t))$. 
\begin{proposition} \label{Paulprop}
At equilibrium, the two first moments of $P$ can be expressed by
\begin{equation}\label{PaulExp}
\E(P) =  \frac{\lambda_{1}^+}{\lambda_{1}^++\lambda_{1}^-}\frac{\lambda_2}{\mu_2}\frac{\lambda_3}{\mu_3}
\end{equation}
\begin{multline}\label{PaulVar}
\var(P)=\E(P)\left(1+\frac{\lambda_{3}}{\mu_{2}+\mu_{3}}\right.\\\left.+\frac{\lambda_1^-\lambda_{2}\lambda_{3}\left(\lambda_{1}^{+}+\lambda_{1}^{-}+\mu_{2}+\mu_{3}\right)}{(\lambda_1^++\lambda_1^-)\left(\mu_{2}+\mu_{3}\right)\left(\lambda_{1}^{+}+\lambda_{1}^{-}+\mu_{2}\right)\left(\lambda_{1}^{+}+\lambda_{1}^{-}+\mu_{3}\right)}\right).
\end{multline}
\end{proposition}
See Paulsson~\cite{Paulsson},  Shahrezaei and Swain~\cite{Swain2008}, Swain et al.~\cite{Swain2002} and Fromion et al.~\cite{Fromion} for example.
\subsection{A Stochastic Model of Protein Production with Autogenous Regulation}

The regulation is done via  proteins  which can inactivate the gene corresponding to the protein.  If there are $P$ proteins at some moment then the gene is activated at a rate proportional to $P$. Compared to the above model, only the first step changes.
\begin{enumerate}
\item The inactive gene is activated at rate $\lambda_1^+$ and inactivated at rate $\lambda_1^-P$ otherwise. 
\end{enumerate}
See Figure~\ref{transfig2}. 
For the sake of simplicity,  we use the same notations $\lambda_1^+$  and $\lambda_1^-$ as for the classical model of protein production instead of $\lambda_{F,1}^{+}$ and $\lambda_{F,1}^{-}$ for example. It should be noted that in our comparisons in Section~\ref{SecDisc}, these quantities are not necessarily the same for these two models. 

The corresponding Markov process is denoted as  $(I_F(t),M_F(t),P_F(t))$, its transitions have the same rate as $(I(t),M(t),P(t))$ except for those concerning the first coordinate.
\[
\begin{cases}
(0,m,p)\rightarrow (1,m,p) \quad\text{ at rate } \lambda_1^+,\qquad& (1,m,p)\rightarrow (0,m,p)\hfill \text{ at rate } \lambda_1^-p.
\end{cases}
\]
As before,  $(I_F,M_F,P_F)$ will denote  random variables whose law is the invariant distribution of the Markov process $(I_F(t),M_F(t),P_F(t))$. The following proposition is the analogue of Proposition~\ref{Paulprop} for the feedback model but with unknown quantities related to the the activity of the gene, $\E(I_F)$, and the correlation of the activity of the gene and the number of mRNAs, $\E\left(I_FM_F\right)$. 
\begin{proposition} \label{FeedProp}
At equilibrium, the  first moment of $P_F$ can be expressed by
\begin{equation}\label{FeedExp}
\E\left(P_F\right)=\E(I_F)\frac{\lambda_2}{\mu_2}\frac{\lambda_3}{\mu_3}.
\end{equation}
\end{proposition}
\begin{proof}
By equality of input and output for  $(M(t))$ and $(P(t))$ at equilibrium, one gets  the relations 
\[
\lambda_2\E\left(I_F\right)=\mu_2\E\left(M_F\right),\quad 
\lambda_3\E\left(M_F\right)=\mu_3\E\left(P_F\right),
\]
and therefore Relation~\eqref{FeedExp}. 
\end{proof}
It does not seem that an expression for $\E(I_F)$ can be obtained, the relation   $\lambda_1^-\E(I_FP_F)=\lambda_1^+(1-\E(I_F))$ of equality of flows for activation/deactivation process introduces  the correlation between $I_F$ and $P_F$. This is in fact the main obstacle to get more insight on the fluctuations of the number of proteins. The next section investigates a scaling where the activation/deactivation phase is much more rapid than the production process of proteins. 

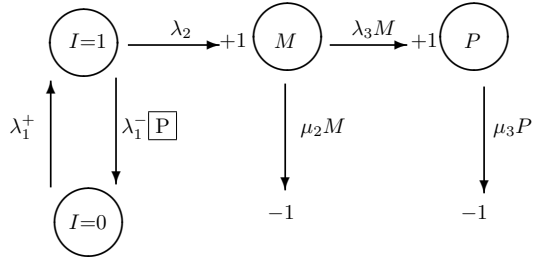
\begin{figure}
\setlength{\unitlength}{2044sp}

\begin{center}
\scalebox{0.8}{
\begin{picture}(7017,3892)(1828,-4594)
{\thicklines
\put(2395,-1328){\circle{1026}}
}%
{\thicklines\put(2465,-4058){\circle{1026}}
}%
\put(2165,-4158){$I{=}0$}
{\thicklines\put(5479,-1238){\circle{1026}}
}%
\put(5279,-1408){$M$}
\put(2165,-1408){$I{=}1$}
\put(3700,-1208){$\lambda_2 $}
\put(6450,-1208){$\lambda_3 M$}
\put(1265,-2700){$\lambda_1^+$}
\put(2965,-2700){$\lambda_1^-\fbox{P}$}
\put(5679,-2700){$\mu_2 M$}
\put(8109,-4000){$-1$}
\put(8600,-2700){$\mu_3 P$}
{\thicklines\put(8309,-1238){\circle{1026}}
}%
\put(8159,-1408){$P$}
{\thicklines\put(3061,-1366){\vector( 1, 0){1390}}
}%
\put(4451,-1400){$+1$}
{\thicklines\put(6121,-1366){\vector( 1, 0){1175}}
}%
\put(7351,-1400){$+1$}
{\thicklines\put(1891,-3526){\vector( 0, 1){1620}}
}%
{\thicklines\put(2881,-1861){\vector( 0,-1){1620}}
}%
{\thicklines\put(5446,-1951){\vector( 0,-1){1620}}
}%
\put(5179,-4000){$-1$}
{\thicklines\put(8461,-1996){\vector( 0,-1){1620}}
}%
\end{picture}
}
\end{center}
\caption{Three Stage Model for Protein Production with Autogenous Regulation}\label{transfig2}
\end{figure}
\section{A Scaling Analysis}\label{scalingsec}
It has been seen in the previous section that, for the feedback mechanism,  an explicit representation of the variance of the number of proteins at equilibrium seems to be difficult to derive. In this section we use the fact that the time scale of the first two steps, activation/deactivation of the gene and production of mRNAs is more rapid than the time scale of protein production. This is illustrated by the fact that the lifetime of an mRNA is of the order of   2mn. whereas the  doubling time  of a bacteria is  around $40$mn giving a lifetime of a protein of the order of one hour.  See Taniguchi et al.~\cite{Taniguchi}, Li and Elf~\cite{LiElf} and Hammar et al.~\cite{Hammar}.  As  will be seen, this assumption simplifies the analysis of the feedback mechanism. We will be able to get an asymptotic explicit expression for the distribution of the number of proteins at equilibrium. 

A (large) scaling parameter $N$ is used to stress the difference of time scale.  When there is a feedback control, an upper index $N$ is added to the variables so that the corresponding Markov process is denoted as  $(X^N_F(t))=(I^N_F(t),M^N_F(t),P^N_F(t))$ on the state space ${\cal S}=\{0,1\}\times\N^2$. The transition rates of the Markov process are given by 
\begin{equation}\label{rates}
\begin{cases}
(0,m,p)\rightarrow (1,m,p) \text{ at rate } \lambda_1^+N, & (1,m,p)\rightarrow (0,m,p)\hfill \text{ at rate } \lambda_1^-Np,\\
(i,m,p)\rightarrow (i,m+1,p) \text{\phantom{at rate}} i\lambda_2N, &(i,m,p)\rightarrow (i,m-1,p) \hfill \mu_2mN,\\
(i,m,p)\rightarrow (i,m,p+1) \text{\phantom{at rate}} \lambda_3 m, & (i,m,p)\rightarrow (i,m,p-1) \hfill \mu_3p.
\end{cases}
\end{equation}
The initial state is constant with $N$ given by $X^N_F(0)=(i_0,m_0,p_0)\in{\cal S}$.

The aim of this section is of proving that the non-Markovian process $(P^N_F(t))$ converges in distribution to a limiting  Markov process $(\overline{P}_F(t))$.  As  will be seen, an averaging principle, proved in the appendix, holds: locally the ``fast'' process $(I^N_F(t),M^N_F(t))$ reaches very quickly some equilibrium depending on the current value of the ``slow'' variable $P^N_F(t)$. It turns out that the equilibrium of this limiting process $(\overline{P}_F(t))$  can be analyzed in detail. The proof of the averaging principle relies on stochastic calculus applied to  Markov processes in the same spirit as in Papanicolau et al.~\cite{PSV} in a Brownian setting, see also Kurtz~\cite{Kurtz:05}.

\subsection*{Notations} Throughout the rest of this paper, we will use the following notations  $\rho_1=\lambda_1^+/\lambda_1^-$ and, for $i=1$, $2$,  $\rho_i={\lambda_i}/{\mu_i}$.

\subsection{Scaling of the Classical Model of Protein Production}
One first states a scaling result for the classical model of protein production. The result being much simpler to prove than the corresponding result, Theorem~\ref{TheoFeed}, for the feedback process, its proof is skipped. 
One denotes by  $(X^N(t))=(I^N(t),M^N(t),P^N(t))$ the corresponding Markov process, its transition rates are the same as for feedback in Relation~\eqref{rates} except for deactivation:
\[
(1,m,p)\rightarrow (0,m,p)\hfill \text{ at rate } \lambda_1^-N. 
\]
The following result shows that, in the limit, the evolution of the number of proteins converges to the time evolution of an $M/M/\infty$ queue. See Chapter~6 of Robert~\cite{Robert} for example. 
\begin{theorem}\label{TheoPaul}
If $X^N(0)=(i_0,m_0,p_0)\in{\cal S}$, the sequence of processes $(P^N(t))$ converges in distribution to a birth and death process $(\overline{P}(t))$ on $\N$ whose respective birth and death rates $(\beta_x)$ and $(\delta_x)$ are given by
\[
\beta_x=\frac{\lambda_3\rho_2\rho_1}{\rho_1+1} \text{  and  } \delta_x=\mu_3x.
\]
The equilibrium distribution of $(\overline{P}(t))$ is a Poisson distribution with parameter  $\rho_1\rho_2\rho_3/(1+\rho_1)$.
\end{theorem}
\begin{proof}
The intuition of this result can be described quickly as follows.  The processes $(I^N(t),M^N(t))$ live on a much faster time scale than $(P^N(t))$ and therefore  reach quickly the equilibrium. When $N$ gets large, the process $(M^N(t))$ is  an $M/M/\infty$ queue with arrival rate $\lambda_2\lambda_1^+/(\lambda_1^++\lambda_1^-)$ and service rate $\mu_2$.  See Chapter~6 of Robert~\cite{Robert} for example. Its equilibrium distribution is therefore  Poisson with parameter $\rho_2\rho_1/(1+\rho_1)$. The process $(P^N(t))$ can then be seen as an  $M/M/\infty$ queue with arrival rate $\lambda_3\rho_2\rho_1/(1+\rho_1)$ and service rate $\mu_3$, i.e. a birth and death process with the  transition rates of the theorem. Its equilibrium is Poisson with parameter  $\rho_1\rho_2\rho_3/(1+\rho_1)$.

The proof of a corresponding result in a more complicated setting, for the production process with feedback, is done below. For this reason the proof of this result is skipped. 
\end{proof}

\begin{figure}
\setlength{\unitlength}{2044sp}

\begin{center}
\scalebox{0.8}{
\begin{picture}(7017,3892)(1828,-4594)
{\thicklines
\put(2395,-1328){\circle{1026}}
}%
{\thicklines\put(2465,-4058){\circle{1026}}
}%
\put(2085,-4158){$I^N_F{=}0$}
{\thicklines\put(5479,-1238){\circle{1026}}
}%
\put(5279,-1408){$M^N$}
\put(2005,-1408){$I^N_F{=}1$}
\put(3700,-1008){$\lambda_2 \fbox{N}$}
\put(6450,-1208){$\lambda_3 M^N$}
\put(805,-2700){$\lambda_1^+\fbox{N}$}
\put(3065,-2700){$\lambda_1^-\fbox{N}P^N$}
\put(5679,-2700){$\mu_2 \fbox{N} M^N$}
\put(8109,-4000){$-1$}
\put(8600,-2700){$\mu_3 P^N$}
{\thicklines\put(8309,-1238){\circle{1026}}
}%
\put(8159,-1408){$P^N$}
{\thicklines\put(3061,-1366){\vector( 1, 0){1390}}
}%
\put(4451,-1400){$+1$}
{\thicklines\put(6121,-1366){\vector( 1, 0){1175}}
}%
\put(7351,-1400){$+1$}
{\thicklines\put(1891,-3526){\vector( 0, 1){1620}}
}%
{\thicklines\put(2881,-1861){\vector( 0,-1){1620}}
}%
{\thicklines\put(5446,-1951){\vector( 0,-1){1620}}
}%
\put(5179,-4000){$-1$}
{\thicklines\put(8461,-1996){\vector( 0,-1){1620}}
}%
\end{picture}}
\end{center}
\caption{Feedback Model with Scaling Parameter $N$}
\end{figure}
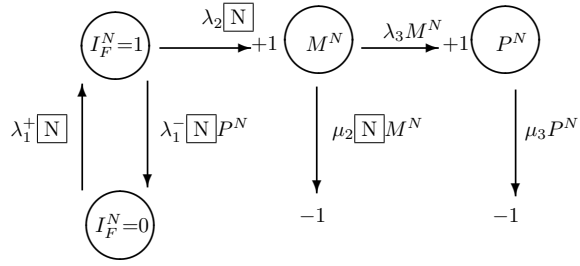

\subsection{Scaling of the Production Process with Feedback }
The following theorem is the main result of this section. As in the case of the classical model of protein production, it relies on the fact that, due to the scaling, the activation/deactivation of the gene and the production of mRNAs occur on a fast time scale so that an averaging principle holds. See below.  Some of the technical results used to establish the following theorem are presented in the Appendix. 

\begin{theorem}[Hill Repression Scheme]\label{TheoFeed}
If $X^N_F(0)=(i_0,m_0,p_0)\in{\cal S}$, the sequence of processes $(P^N_F(t))$ converges in distribution to a birth and death process $(\overline{P}_F(t))$ on $\N$ whose respective birth and death rates $(\beta_x)$ and $(\delta_x)$ are given by
\[
\beta_x=\frac{\lambda_3\rho_2\rho_1}{\rho_1+x} \text{  and  } \delta_x=\mu_3x,
\]
with $\rho_1=\lambda_1^+/\lambda_1^-$ and $\rho_2={\lambda_2}/{\mu_2}$.
\end{theorem}
\begin{proof}
If $f$ is a function on $\N$ with finite support then
\begin{multline*}
V^N_f(t)\stackrel{\text{def.}}{=} f(P^N_F(t))-f(p_0)\\
-\int_0^t \lambda_3 M^N_F(u) \Delta^+(f)(P^N_F(u))\,\diff u-\int_0^t \mu_3 P^N_F(u) \Delta^-(f)(P^N_F(u))\,\diff u,
\end{multline*}
is a local martingale. See Rogers and Williams~\cite{Rogers} for example. The operators $\Delta^+$ and $\Delta^-$ are defined as follows, for a real-valued function $f$ on $\N$,
\[
\Delta^+(f)(x)=f(x+1)-f(x) \text{ and } \Delta^-(f)(x)=f(x-1)-f(x), \quad x\in\N. 
\]
With a similar method as in the proof of Assertion~1) of Lemma~\ref{lem1} in the appendix and by using the criterion of the modulus of continuity, see Theorem 7.2 page~81 of Billingsley~\cite{Billingsley}, it is easy to show that the two processes
\[
\left(\int_0^t \lambda_3 M^N_F(u) \Delta^+(f)(P^N_F(u))\,\diff u\right) \text{ and }  \left(\int_0^t \mu_3 P^N_F(u) \Delta^-(f)(P^N_F(u))\,\diff u\right)
\]
are tight. Because of the tightness of $(P^N_F(t))$ of Proposition~\ref{lem2} of the appendix, one can take $(N_k)$ a subsequence such that the process
\[
\left(P^{N_k}_F(t),\int_0^t \lambda_3 M^{N_k}_F(u) \Delta^+(f)(P^{N_k}_F(u))\diff u, \int_0^t \mu_3 P^{N_k}_F(u) \Delta^-(f)(P^{N_k}_F(u))\diff u\right)
\]
converges in distribution.

Let $(\overline{P}_F(t))$ be a possible limit of $(P_F^{N_k}(t))$, then by continuity of the mapping
\[
(z(t)) \mapsto \left(\int_0^t  z(u) \Delta^-(f)(z(u))\,\diff u\right)
\]
on ${\cal D}([0,T])$ endowed with the Skorohod topology then, for the convergence in distribution
\begin{multline*}
\lim_{k\to+\infty}\left(P_F^{N_k}(t), \int_0^t  P_F^{N_k}(u) \Delta^-(f)(P_F^{N_k}(u))\,\diff u\right)
\\=\left(\overline{P}_F(t), \int_0^t  \overline{P}_F(u) \Delta^-(f)(\overline{P}_F(u))\,\diff u\right).
\end{multline*}
For $t\leq T$, by using the definition of $\Lambda^N$ and of ${\cal E}_T$ in Section~\ref{OccMeas} of the Appendix, one has the relation
\[
\int_0^t  M_F^{N_k}(u) \Delta^+(f)(P_F^{N_k}(u))\,\diff u=
\int_{{\cal E}_T}  m\Delta^+(f)(p)\mathbbm{1}_{[0,t]}(u)\,\Lambda^{N_k}(\diff z),
\]
hence, by Proposition~\ref{occmes} of Appendix, for the convergence in distribution
\begin{multline*}
\lim_{k\to+\infty} \int_{{\cal E}_T}  m\Delta^+(f)(p)\mathbbm{1}_{[0,t]}(u)\,\Lambda^{N_k}(\diff z),
\\=\int_0^t \sum_{p\in\N} \Delta^+(f)(p)\sum_{(i,m)\in\{0,1\}\times \N} m\ell_u(i,m,p)\diff u,
\\=\int_0^t \sum_{p\in\N} \Delta^+(f)(p)\frac{\lambda_1^+}{\lambda_1^++\lambda_1^-p}\frac{\lambda_2}{\mu_2}\nu_u(p)\,\diff u
\end{multline*}
by Relation~\eqref{eqEM} of Proposition~\ref{occmes} of the Appendix. By convergence of the sequence $(\Lambda^{N_k})$ this last expression can be expressed as
\begin{multline*}
\left(\int_0^t \sum_{p\in\N} \Delta^+(f)(p)\frac{\lambda_1^+}{\lambda_1^++\lambda_1^-p}\nu_u(p)\,\diff u\right)\\=
\lim_{k\to+\infty} 
\left(\int_0^t  \Delta^+(f)(P^{N_k}_F(u))\frac{\lambda_1^+}{\lambda_1^++\lambda_1^-P^{N_k}_F(u)}\,\diff u\right)
\\
\stackrel{\text{dist.}}{=} \left(\int_0^t  \Delta^+(f)(\overline{P}_F(u))\frac{\lambda_1^+}{\lambda_1^++\lambda_1^-\overline{P}_F(u)}\,\diff u\right)
\end{multline*}
for the convergence in distribution. 

For $0\leq s\leq t$, the characterization of a Markov process as the solution of a martingale problem gives the relation 
\begin{multline*}
\E\left(f(P^N_F(t))-f(P^N_F(s))-\int_s^t \lambda_3 M^N_F(u) \Delta^+(f)(P^N_F(u))\,\diff u\right.\\
\left.\left.-\int_s^t \mu_3 P^N_F(u) \Delta^-(f)(P^N_F(u))\,\diff u\right|{\cal F}_s\right)=0,
\end{multline*}
from which we deduce the identity
\begin{multline*}
\E\left(f(\overline{P}_F(t))-f(\overline{P}_F(s))-\int_s^t  \lambda_3\frac{\rho_1\rho_2}{\rho_1+\overline{P}_F(u)}\Delta^+(f)(\overline{P}_F(u)) \right.\\
\left.\left.-\int_s^t \mu_3 \overline{P}_F(u) \Delta^-(f)(\overline{P}_F(u))\,\diff u\right|{\cal F}_s\right)=0.
\end{multline*}
See Theorem II.2.42 of Jacod and Shiryaev~\cite{Jacod}. 
Consequently, a possible limit is the solution of the martingale problem associated to the birth and death process with birth rate $(\beta_x)$ and death rate $(\delta_x)$ and with initial state in $p_0$.  One gets therefore  the desired convergence in distribution of $(P^N_F(t))$. The theorem is proved. 
\end{proof}
There exist cases where the autoregulation is not achieved by the regulated protein but by a complex of this protein, e.g by a dimer (2 copies of the protein) or a tetramer (4 copies) to cite few examples.  In order to handle such cases, it is necessary to add to the gene expression model,  a preliminary step describing the reaction scheme of the complex formation based on the law of mass action.  In general, the dynamics involved in the reaction scheme are (very) rapid compared to the other processes of the gene expression and leads, by a singular perturbation like argument, to represent in case of deterministic model  the  rate of production of mRNAs as a non-linear function of protein concentration. Furthermore, when the reaction scheme possesses  suitable properties,  a Hill like repression function could also be  obtained. See Weiss~\cite{Wei:97} for details.  In the stochastic context, that leads to introduce a suitable scaling factor in the dynamics of the complex formation and to extend the previous derivation in the previous theorem to  Hill functions, $x\mapsto a/(b+x^n)$, with order $n$ greater than 1.

The next section analyzes, in this limiting regime, the fluctuations of the number of proteins at equilibrium. 
\section{Fluctuations of the Number of Proteins}\label{SecInv}
This section is devoted to the analysis of the equilibrium of the asymptotic process $(\overline{P}_F(t))$ of Theorem~\ref{TheoFeed} describing the evolution of the number of proteins with  feedback. We start with a classical result for birth and death processes. 
\begin{proposition}
The invariant distribution $\pi_F$ of the birth and death process $(\overline{P}_F(t))$ of Theorem~\ref{TheoFeed} is given by
\[
\pi_F(x)=\frac{1}{Z} \frac{(\rho_2\rho_3)^x}{x!} \prod_{i=0}^{x-1} \frac{\rho_1}{\rho_1+i}, \qquad x\in\N,
\]
where $\rho_1=\lambda_1^+/\lambda_1^-$, $\rho_i=\lambda_i/\mu_i$ for $i=1$, $2$ and $Z$ is the normalization constant. 
\end{proposition}
The expression of $\pi_f$ is explicit but with a normalization constant which is not simple. The constant $Z$ can be expressed in terms of  hypergeometric functions. See Abramowitz and Stegun~\cite{Abram} for example. 
Even if we can get a numerical evaluation of the average and of the variance of $\pi_F$, it is much more awkward to get some insight on the dependence of these quantities with respect to some of the parameters like $\rho_2$ or $\rho_3$ for example. In the following we give an asymptotic description of the ratio of the variance and the mean of the number of proteins at equilibrium when the value of the quantity $\rho_1\rho_2\rho_3$ is large. In a biological context the numerical value of this parameter is not always large but this limit results sheds some light on the qualitative behavior of the auto-regulation mechanism. See Corollary~\ref{corol1} for example. 
A Laplace method is in particular used to investigate the asymptotic behavior of the first two moments of $\pi_F$. 

Theorem~\ref{TheoPaul} shows that the distribution of the process $(P(t))$ at equilibrium is Poisson with parameter  $\E(P(t))=x_\rho{=}\rho_1\rho_2\rho_3/(1+\rho_1)$. In particular, 
one has the relation $\var(P(t))=\E(P(t))$. In the rest of this section, we will be interested in the corresponding quantity for the feedback process.

For $\eta>0$ and $\rho>0$, denote by $\nu_\rho$ the probability distribution on $\N$ defined by
\begin{equation}\label{nurho}
\nu_{\rho,\eta}(k)=\frac{1}{Z_\rho} \frac{\rho^k}{k!}\prod_{i=1}^k \frac{1}{\eta+i}
=\frac{1}{Z_\rho}\exp\left(\sum_{i=1}^k \log\left(\frac{x}{i(\eta+i)}\right)\right),
\end{equation}
where $Z_\rho$ is the normalization constant. It is easily seen that $\pi_F$ is $\nu_{\rho,\eta}$ with $\rho=\rho_1\rho_2\rho_3$ and $\eta=\rho_1-1$. 
\begin{proposition}\label{cvdistA}
If, for $\rho>0$ and $\eta>-1$,  $A_\rho$ is a random variable with distribution $\nu_{\rho,\eta}$ defined by Equation~\eqref{nurho}, then for the convergence in distribution
\[
\lim_{\rho\to +\infty} \frac{A_\rho-a_\rho}{\sqrt{a_\rho}}=
{\cal N}\left(0,1/\sqrt{2}\right),
\]
where $a_\rho=\left(\sqrt{\eta^2+4\rho}-\eta\right)/2$ and ${\cal N}\left(0,1/\sqrt{2}\right),$ is a centered Gaussian random variable with variance $1/2$. 
In particular, for the convergence in distribution,
\[
\lim_{\rho\to+\infty} \frac{A_\rho}{\sqrt{\rho}}=1.
\]
\end{proposition}
\begin{proof}
If $\phi$ is a bounded function on $\R$, denote
\[
\Delta_\rho(\phi)\stackrel{\text{def.}}{=}\frac{1}{\sqrt{a_\rho}}\sum_{k=0}^{+\infty} \phi\left(\frac{k-\lceil a_\rho\rceil}{\sqrt{a_\rho}}\right)\exp\left(\sum_{i=\lceil a_\rho\rceil}^k \log\left(\frac{\rho}{i(\eta+i)}\right)\right),
\]
the definition of $\nu_{\rho,\eta}$ gives that
\begin{equation}\label{eqr3}
\E\left(\phi\left(\frac{A_\rho-\lceil a_\rho\rceil }{\sqrt{a_\rho}}\right)\right) = \frac{\Delta_\rho(\phi)}{\Delta_\rho(1)}
\end{equation}
Fix $\phi$  some continuous function with compact support on $[-K_0,K_0]$ for some $K_0{>}0$. Since $a_\rho$ is the solution of the equation $a_\rho(\eta{+}a_\rho){=}\rho$,  a change of variable gives the relation
\[
\Delta_\rho(\phi){=}\frac{1}{\sqrt{a_\rho}}\sum_{k=-\lfloor K_0\sqrt{a_\rho}\rfloor}^{\lceil K_0\sqrt{a_\rho}\rceil} \hspace{-5mm}\phi\left(\frac{k}{\sqrt{a_\rho}}\right)\exp\left(\sum_{i=0}^k \log\left(\frac{a_\rho(\eta+a_\rho)}{(i+\lceil a_\rho\rceil)(\eta+\lceil a_\rho\rceil+i)}\right)\right).
\]
The uniform estimation
\begin{multline*}
\sum_{i=0}^k \log\left(\frac{a_\rho(\eta+a_\rho)}{(i+\lceil a_\rho\rceil)(\eta+\lceil a_\rho\rceil+i)}\right)\\=\int_{0}^{k} \log\left(\frac{a_\rho(\eta+a_\rho)}{(u+\lceil a_\rho\rceil)(\eta+\lceil a_\rho\rceil+u)}\right)\,\diff u+O\left(\frac{1}{\sqrt{a_\rho}}\right)
\end{multline*}
for all $k\in\Z$ with $|k|\leq K_0\sqrt{a_\rho}$ and the fact that $\phi$ has a compact support give that the quantity $\Delta_\rho(\phi)$ is equivalent to 
\begin{align*}
&\frac{1}{\sqrt{a_\rho}}\sum_{k=-\lfloor K_0\sqrt{a_\rho}\rfloor}^{\lceil K_0\sqrt{a_\rho}\rceil} \hspace{-5mm}\phi\left(\frac{k}{\sqrt{a_\rho}}\right)\exp\left(\int_{0}^{k} \log\left(\frac{a_\rho(\eta{+}a_\rho)}{(u{+}\lceil a_\rho\rceil)(\eta{+}\lceil a_\rho\rceil+u)}\right)\,\diff u\right)\\
  &=\frac{1}{\sqrt{a_\rho}}\sum_{k{=}{-}\lfloor K_0\sqrt{a_\rho}\rfloor}^{\lceil K_0\sqrt{a_\rho}\rceil} \hspace{-5mm}\phi\left(\frac{k}{\sqrt{a_\rho}}\right)\\
  &\hspace{1cm}\times \exp\left(\int_{0}^{k/\sqrt{a_\rho}} \hspace{-4mm}\sqrt{a_\rho}\log\left(\frac{a_\rho(\eta+a_\rho)}{(u\sqrt{a_\rho}{+}\lceil a_\rho\rceil)(\eta{+}\lceil a_\rho\rceil{+}u\sqrt{a_\rho})}\right)\diff u\right).
\end{align*}
Again, with the uniform estimation
\[
 \sqrt{a_\rho}\log\left(\frac{a_\rho(\eta+a_\rho)}{(u\sqrt{a_\rho}+\lceil a_\rho\rceil)(\eta+\lceil a_\rho\rceil+u\sqrt{a_\rho})}\right)=-2u+O\left(\frac{1}{\sqrt{a_\rho}}\right),
\]
for $u$ in some fixed  finite interval, one gets that
\[
\Delta_\rho(\phi){\sim}\frac{1}{\sqrt{a_\rho}}\sum_{k=-\lfloor K_0\sqrt{a_\rho}\rfloor}^{\lceil K_0\sqrt{a_\rho}\rceil}\hspace{-3mm} \phi\left(\frac{k}{\sqrt{a_\rho}}\right)\exp\left({-}2\int_{0}^{k/\sqrt{a_\rho}}u\,\diff u\right){\sim}\int_{-\infty}^{+\infty} \hspace{-5mm}\phi\left(v\right)e^{-v^2}\,\diff v.
\]
With similar estimations for $\Delta_\rho(1)$ (which imply in fact the tightness of the random variables $(A_\rho-\lfloor a_\rho\rfloor)/\sqrt{a_\rho}$) and Relation~\eqref{eqr3},  the proposition is proved. 
\end{proof}
\begin{corollary}[Asymptotic Number of Proteins with Regulation]\label{corol1}
If $\overline{P}_F$  is a random variable with distribution $\pi_F$ then, for the convergence in distribution
\begin{equation}\label{varFeed}
\lim_{\rho_2\rho_3 \to+\infty} \frac{\E(\overline{P}_F)}{\sqrt{\rho_1\rho_2\rho_3}} =1
\text{ and  }
\lim_{\rho_2\rho_3\to+\infty} \frac{\var(\overline{P}_F)}{\E(\overline{P}_F)}=\frac{1}{2}. 
\end{equation}
Furthermore, for the convergence in distribution, 
\[
\lim_{\rho_2\rho_3\to +\infty} \frac{\overline{P}_F-a_\rho}{\sqrt{a_\rho}}=
{\cal N}\left(0,1/\sqrt{2}\right),
\]
where $a_\rho=\left(\sqrt{(\rho_1-1)^2+4\rho_1\rho_2\rho_3}-\rho_1+1\right)/2$.
\end{corollary}
\noindent
The equivalent of  Relation~\eqref{varFeed} for the scaling of the  classical model of protein production is
\[
\E(P)=\frac{\rho_1}{1+\rho_1}\rho_2\rho_3
\text{ and  }
\frac{\var(P)}{\E(P)}=1.
\]
by Theorem~\ref{TheoPaul}. it shows that a feedback mechanism reduces the variance of the number of proteins in this limiting regime by a factor $2$ for the ratio of the second moment and  the first moment. 
\section{Discussion}\label{SecDisc}
In this section, other aspects of regulation of protein production are discussed via simulations in a plausible biological context whose parameters are going to be defined.  Simulation follows the models in Section 2.2 and simulates the variables $I_F$, $M_F$, and $P_F$, not their scaling limits.
\subsection{\label{biopar}Numerical Values of Biological Parameters}
For the model with feedback, there are six parameters to determine. By using the literature one can estimate the common orders of magnitude of these parameters in a biological context. We therefore propose a  set of parameters  corresponding  to an ``ordinary'' gene.   

\begin{enumerate}
\item Gene regulation. The parameter $\lambda_{1}^{-}$ gives the rate at which a given protein reaches its own promoter. It has been shown that this motion combines a three-dimensional diffusion in the cytoplasm and one-dimensional sliding along the DNA, see Halford~\cite{Halford}. 

Experiments on the lac repressor, using live-cell single-molecule imaging techniques, show that this time is of the order of 5 min, see Li and Elf~\cite{Li} and Hammar et al.~\cite{Hammar}. For this reason we will take $\lambda_{1}^{-}=3.3\times10^{-3}\,\mbox{\ensuremath{s^{-1}}}$.

The parameter $\lambda_{1}^{+}$ can be quite variable, depending on the affinity of the protein to the DNA sequence, we set $\lambda_{1}^{+}=1\,\mbox{\ensuremath{s^{-1}}}$.

\item mRNAs. The lifetime of an mRNA is $\mu_{2}^{-1}\simeq4\,\textup{min}$, see Taniguchi et al.~\cite{Taniguchi}. When the gene expression is always active (corresponding to the case where our variable $I$ remains equals to $1$), there is an average of $2$ messengers, that is to say $\lambda_2^{-1} = \mu_2^{-1} / m = 120\,s$ which gives $\lambda_{2}=8.3\times10^{-3}\,s^{-1}$.

\item Proteins. A doubling time for the cell of $t_{1/2}\simeq40$\,min gives a protein decay of around one hour. For this reason one takes $\mu_{3}=\log2/t_{1/2}=2.8\times 10^{-4}\,s^{-1}$ for the rate of protein decay. It is assumed that  a give type of  protein that is produced in $p=300$ copies when the gene expression is always active. From one messenger, a protein should be produced in a duration of time of the order of $\lambda_3^{-1}=m\times  \mu_3^{-1} / p$ which gives $\lambda_3= 4\times  10^{-2} \,s^{-1}$.

\end{enumerate}
These parameters may correspond to an ``ordinary bacterial'' gene: in a E. Coli genome of 4300 genes, there are around $3.6\times 10^6$ proteins and $1.4\times 10^3$ mRNAs  per gene, see Table~1 of Chapter~3 of Neidhardt~\cite{Neidhardt}, the number of messengers and proteins is of the order of magnitude of our numerical estimation of the parameters. 

\subsection{Impact of Autogenous Regulation on Gene Expression}
\begin{figure}
\begin{center}
\includegraphics[scale=0.4]{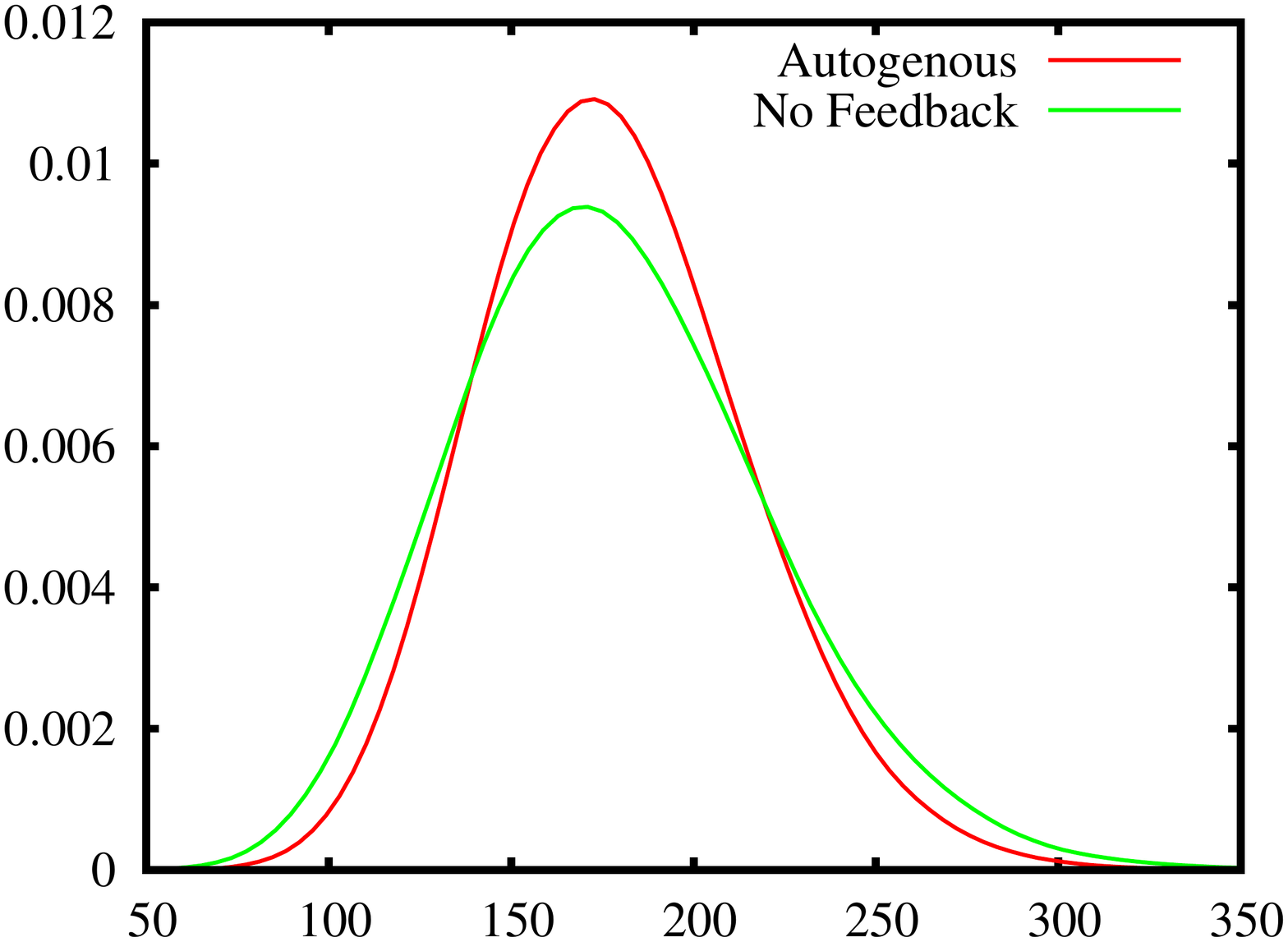}
\end{center}
\caption{\label{distrib_comp} Simulations: Protein distribution with and without autogenous regulation with a fixed mean number of proteins of $178$. }
\end{figure}

We have compared  two mechanisms: the classical model without regulation and the autogenous regulation process. The mean number of proteins is the same as well as the mean number of mRNAs produced $\E(M)=\E(M_F)$. Parameters $\lambda_1^+$ and $\lambda_1^-$ are adapted in the classical model to fulfill these conditions. The other parameters are as defined in the previous section. 

The comparison is shown in Figure~\ref{distrib_comp}. The mean number of proteins is $178$, as can be seen that the curve for the autogenous regulation is slightly more concentrated around the mean but not that much. The values of the corresponding standard deviations are not really different  $\sqrt{\var(P)}{=}42.2$ and $\sqrt{\var(P_F)}{=}35.8$. The impact of the autogenous regulation on the variability of the number of proteins is non-trivial but  not really spectacular for the set of parameters associated to a ``typical'' gene.  This is significantly less than the performances of the limiting mathematical model studied in Section~\ref{scalingsec}. The main reason seems to be that  the scaling parameter is not, in some cases, sufficiently large to have a reasonable accuracy with the limit given by the convergence result of Theorem~\ref{TheoFeed}.

\subsection{The Limiting  Scaling Regime as a Lower Bound} 
Roughly speaking, Theorem~\ref{TheoPaul} and Corollary~\ref{corol1} give that for $N$ and $\rho_2\rho_3$ large, then the ratio $\var(P_F^N)/\E(P_F^N)$ converges to $1/2$. In Figure~\ref{ScaleFig}, one considered a simulation with  fixed product $\rho_2\rho_3=71.43$ with $N$ varying. The interesting feature is that the ratio is decreasing with $N$, this suggests that the variance of the limit of the scaling procedure should provide a lower bound for the variance of the real model. We have not been able to show rigorously this phenomenon.  For  $N=250$, the value of the ratio $\var(P_F^N)/\E(P_F^N)=.7964$ which is quite far from its limiting value $1/2$ given by Corollary~\ref{corol1}. This  can be explained by the fact that the quantities $N$ and $\rho_2\rho_3$ are not very large. 
\begin{figure}
\begin{center}
\includegraphics[scale=0.3]{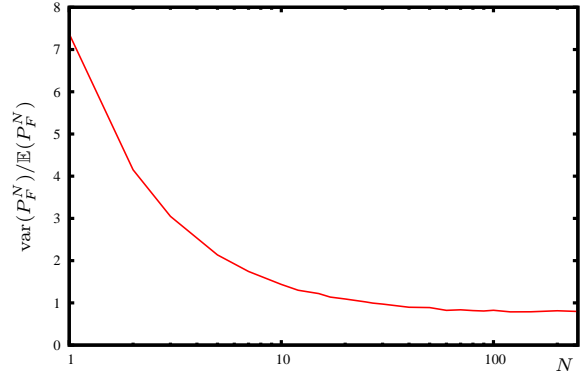}
\put(-25,25){$\scriptstyle{N}$}
\put(-230,70){\rotatebox{90}{$\scriptstyle{\var(P_F^N)/\E(P_F^N)}$}}
\end{center}
\caption{Simulations: Evolution of the ratio $\var(P_F^N)/\E(P_F^N)$ as a function of $N$. }\label{ScaleFig}
\end{figure}

\subsection{Regulation of the Production Process on mRNAs}
The regulation on the gene has the effect of an ON/OFF mechanism. When the gene is active, it is producing mRNAs at full speed and no mRNA is produced  when it is inactive. This suggests that the production of proteins follows roughly  the same pattern:  steady production rate at some instants and little is produced otherwise. This scheme can consequently increase the variability of the production process of proteins. A possible idea to reduce the variance due to the  activation/inactivation of the gene is to transfer the activation/inactivation process at the level the mRNAs. This possibility is investigated in this section.  Each mRNA can be inactivated by a protein at rate $\lambda_2^-$, in this state it cannot produce proteins.  An inactivated mRNAs becomes active at rate $\lambda_2^+$. In this way the production process can, hopefully, be  modulated more smoothly by playing on the inactivation of a fraction of the mRNAs. 
In this way at time $t$,  if the number of active [resp. inactive] mRNAs is $M(t)$ [resp. $M^*(t)$], the process  $(M(t),M^*(t),P(t))$ is  Markov with transition rates, for $(m,m^*,p)\in\N^3$, 
\[
\begin{cases}
(m,m^*,p)\rightarrow (m+1,m^*,p) \text{ at rate } \lambda_2, \\
(m,m^*,p)\rightarrow (m-1,m^*+1,p) \text{ at rate } \lambda_2^-mp, \\
 (m,m^*,p)\rightarrow (m+1,m^*-1,p)\hfill \text{ at rate } \lambda_2^+m^*,
\end{cases}
\]
the other transitions are as before, active of inactive mRNAs die at rate $\mu_2$ and proteins are produced at rate $\lambda_3m$ and die at rate $\mu_3$. 

To compare the two regulation processes, either on the gene or on mRNAs, simulations have been done with the following constraints: the average number of proteins is fixed around 1400. To have a fair comparison, we add the constraint that the number of mRNAs produced should be the same in all simulations. The numerical values have been estimated by using similar methods as in section~\ref{scalingsec} but for this setting.  Experiment~(3) considers the case of an average  lifetime of an mRNA of 40mn, if this is far from a ``normal'' biological setting, as it will be seen, this scenario has the advantage of stressing the importance of this parameter in this configuration. 

\subsection*{Numerical Values of Parameters}
\begin{enumerate}
\item Regulation on the gene. 

\medskip
\begin{tabular}{|l|l|l|l|l|l|}\hline
$\lambda_1^+$ & $\lambda_1^-$& $\lambda_2$& $\mu_2$& $\lambda_3$& $\mu_3$\\\hline
0.21''&5'&12''&4'&25''&1h.\\\hline
\end{tabular}
\item[] 
\item Regulation on mRNAs (I). \\
For this experiment, the expected lifetime of an mRNA is twice the corresponding value of case~(1). 

\medskip
\begin{tabular}{|l|l|l|l|l|l|}\hline
 $\lambda_2$&$\lambda_2^+$ & $\lambda_2^-$& $\mu_2$& $\lambda_3$& $\mu_3$\\\hline
23''& 2''& 45'&8'&25''&1h.\\\hline
\end{tabular}
\item[] 
\item Regulation on mRNAs (II). \\
For this second experiment on the regulation of mRNAs, the expected lifetime of an mRNA is $10$ times than in case~(1). 

\medskip
\begin{tabular}{|l|l|l|l|l|l|}\hline
 $\lambda_2$&$\lambda_2^+$ & $\lambda_2^-$& $\mu_2$& $\lambda_3$& $\mu_3$\\\hline
23.8''& 2''&45'&40'&25''&1h.\\\hline
\end{tabular}
\end{enumerate}

\medskip

\subsection*{Results of the Experiments}

Table~\ref{tab1} shows that the mean number of mRNAs produced per unit of time is essentially the same in all  experiments as well as the mean number of active mRNAs. It should be noted the impact of regulation on  mRNAs for the standard deviation of the number of proteins when the mean life time is $8$mn is not really significant ($10$\% gain) than the regulation on the gene. When the mean lifetime is $40$mn the improvement, $36$\%, of the standard deviation  becomes significant, showing that in this case the production process is ``smoothed'' by this mechanism. The three distributions of the number of proteins of these experiments are presented in  Figure~\ref{fig3}. 
\begin{table}[ht]
\scalebox{0.8}{
\begin{tabular}{|l|l|l|l|}
\hline
Regulation on& Gene & mRNAs/8mn&mRNAs/40mn\\\hline
Mean Nb of  mRNAs & 10.33&19.74 &99.04 \\\hline
Mean Nb of Active mRNAs & 10.33&9.77 &9.81 \\\hline
Mean Nb of Proteins & 1403.63& 1400.29 &1403.36 \\\hline
 Standard Deviation of Nb of Proteins& 92.66& 84.22& 59.04\\\hline
\end{tabular}}
\medskip

\caption{Comparison of Regulation Processes on Gene or on mRNAs with Different Lifetimes}\label{tab1}
\end{table}

\begin{figure}[ht]
\scalebox{1.0}{
\begin{picture}(500,200)(0,20)
\put(212,185){${\scriptstyle \mu_2^{-1}{=}40\text{mn}}$}
\put(215,195){${\scriptstyle \mu_2^{-1}{=}8\text{mn}}$}
\put(215,205){${\scriptstyle \mu_2^{-1}{=}4\text{mn}}$}
\includegraphics[scale=0.4]{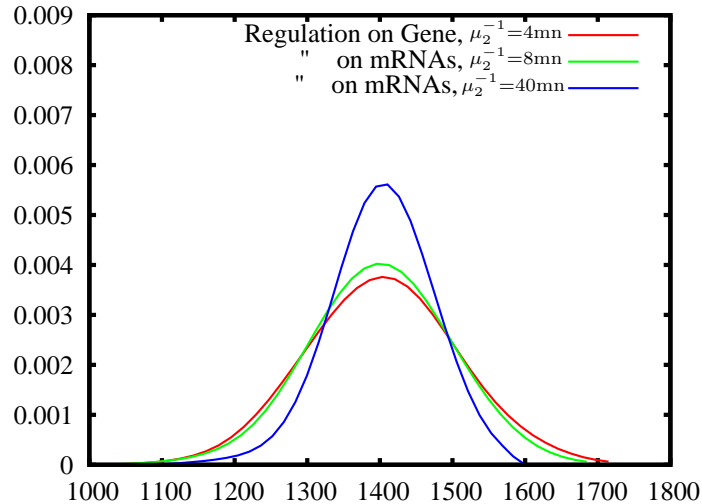} 
\end{picture}}
\caption{Simulations: Probability Distribution of the Number of Proteins with Regulation on Gene or on mRNAs,  $\mu_2^{-1}$ is the  average lifetime of an mRNA. The average number of proteins is $1400$. }\label{fig3}
\end{figure}

\subsection{Impact of Feedback on Frequency}\label{FreqSec}
In this section, we study the nature of the fluctuations of the number of proteins at equilibrium from the point of view of  signal processing or automatic control. The aim of a feedback is often of changing the nature of the signal, attenuating disturbances by reducing, for instance, high frequencies. In these cases, spectral analysis gives a characterization of the nature of changes.

By analogy, we  consider our model as a system that has to achieve a command (the  production of a given mean number of proteins) and where the resulting signal $P(t)$  or $(P_F(t)$ is altered by some noise. In this framework, one can study if the effect of the feedback has an impact on the signal, by rejection of some frequency ranges.

To do so,  consider the signals $(P(t))$ and $(P_F(t))$ of  two simulations with or without autogenous regulation. The analysis of these signals is done by estimating the power spectral density, that describes the spectral characteristics of stochastic process. We estimate the power spectral density for each signal, using classical estimator of smoothed periodogram. See George et al.~\cite{George} and Chapter 10 of Miller et al.~\cite{Miller} for example.

The result is shown in Figure~\ref{fig:Spectrum}. Both spectra seem to represent a low-pass filter with a cut off frequency in the order of magnitude of the dilution factor $\mu_{3}=2.8\times 10^{-4}\,s^{-1}$.  The two power spectral densities do not seem to exhibit significant differences. The feedback has therefore no noticeable effect in terms of reduction of  frequency disturbances.

\begin{figure}[!ht]
\scalebox{0.3}{ \includegraphics{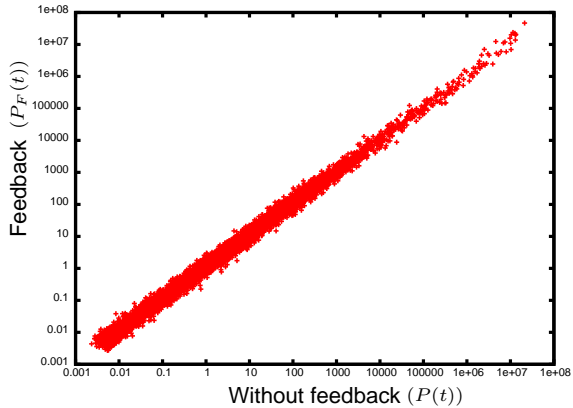}}
\put(-70,17){$\scriptstyle{(P(t))}$}
\put(-220,120){\rotatebox{90}{$\scriptstyle{(P_F(t))}$}}
\caption{Power spectral density estimation of signals with and without regulation}\label{fig:Spectrum}
\end{figure}

\subsection{\label{Vers}Versatility of the Protein Production  Process}
This section is devoted to the impact of autogenous regulation on another aspect of protein production.  Up to now, we have considered the production process of proteins at equilibrium, by assuming that the production rate of a given protein has to be fixed. It may happen nevertheless that, due to an external stress, such as antibiotics, DNA damage by UV, see Camas et al.~\cite{Camas}, or nutriment absorption, see Schleif~\cite{Schleif}, the cell has to change rapidly its production rate to quickly produce a large amount of proteins for example. The affinity of the transcription factor for the promoter of the gene can be adapted for that purpose.  Conversely, when the external stress disappears, the production of the protein has to be quickly reduced to minimize the consumption of resources.

We consider the situation when the two production processes, with and without autogenous regulation, give the same average output of proteins at equilibrium. Two cases are investigated: when the initial number of proteins is below the value equilibrium, see Figure~\ref{FigVersA}, or above this value, see Figure~\ref{FigVersB}. As it can be seen, the autogenous production process converges more rapidly to equilibrium in both cases. 
Our simulations show that when the initial value is 290, the autogenous production process is 40\% faster  than the process without feedback to reach the level 1300 (the equilibrium is at 1400 in this case). A similar result holds in the other case.

These interesting properties are related to the modulation of the gene activity. In the experiment of Figure~\ref{FigVersA}, for the autogenous process the rate of activity of the gene is of the order of  50\% at the beginning and it is only of the order of 0.1 later at equilibrium. Without regulation this rate is constant throughout the simulation. This explains the ``fast start'' of the autogenous process. An analogous explanation holds for the experiment of Figure~\ref{FigVersA}, in the autogenous process. The gene is rapidly switched off due to the large number of proteins, thereby decreasing rapidly the number of proteins. This is consistent with experiments described in Camas et al.~\cite{Camas} and especially Rosenfeld et al.~\cite{Rosenfeld} where the improvement has been estimated at 80\% in some cases.

\begin{figure}[!ht]
\scalebox{0.3}{\includegraphics{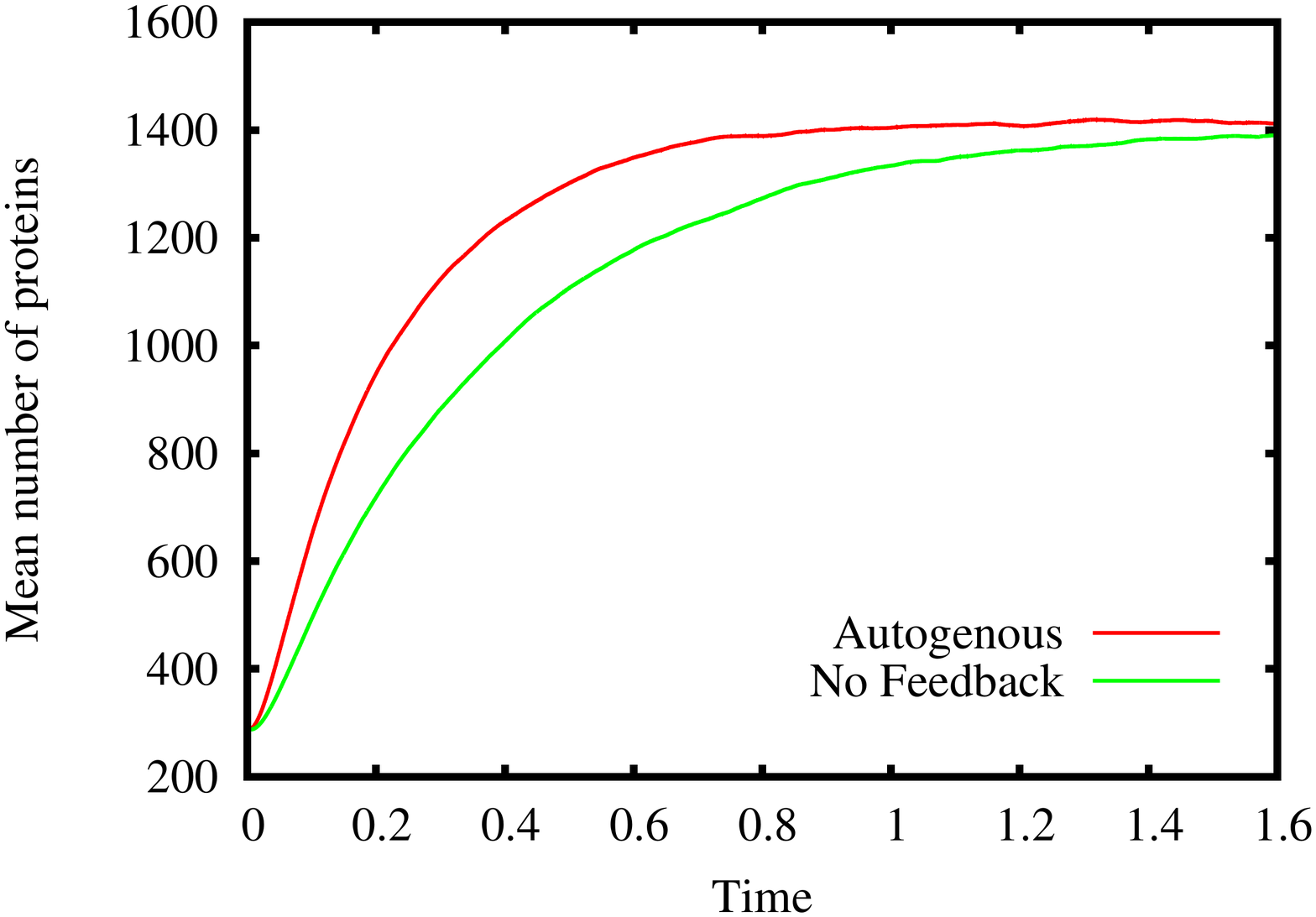}}
\caption{Simulations: Evolution of the Mean Number of Proteins: Initial Point at 290, equilibrium at 1400. Time scale: $\times 10^4$ sec.}\label{FigVersA}
\end{figure}

\begin{figure}[!ht]
\scalebox{0.3}{ \includegraphics{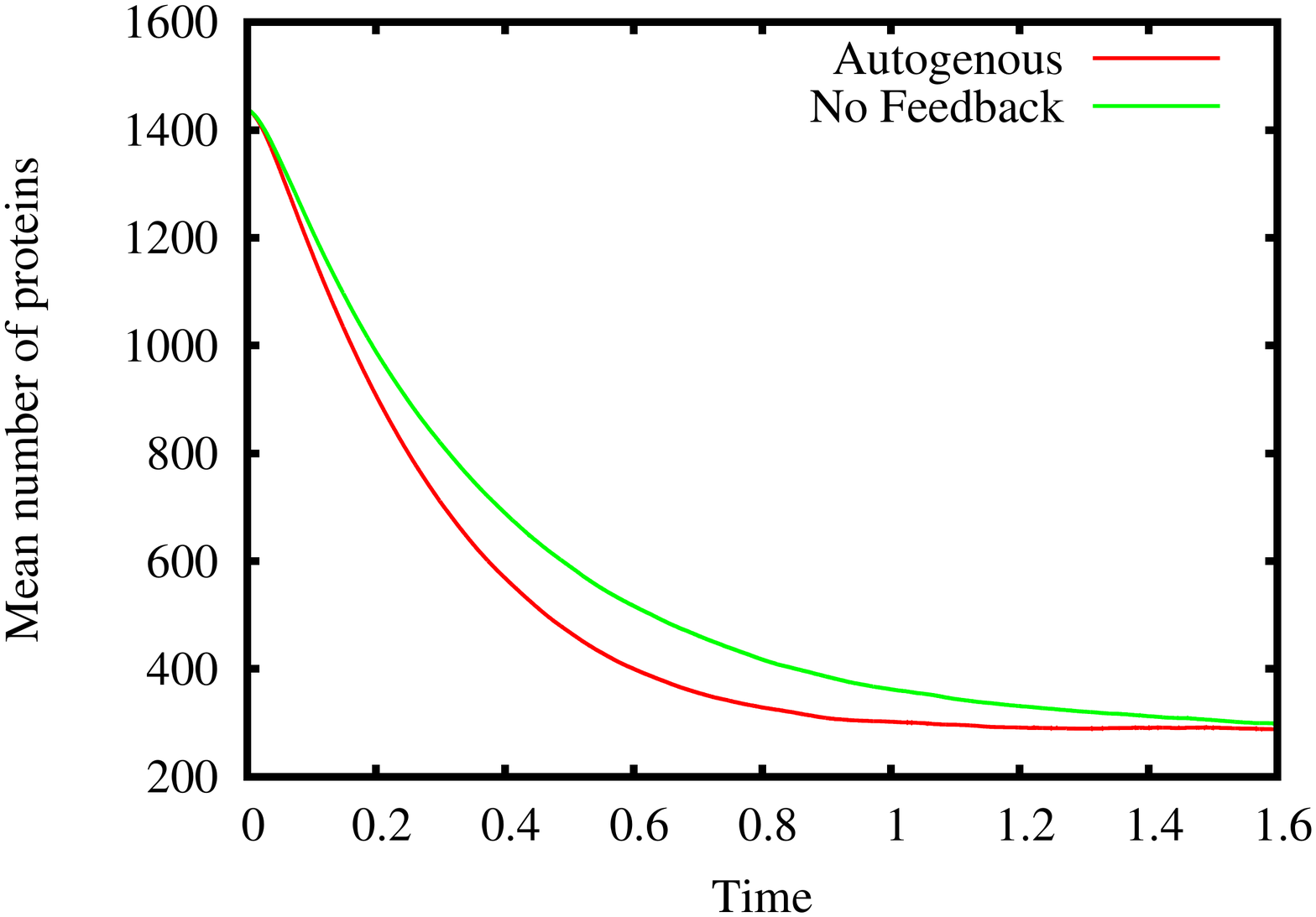}}
\caption{Simulations: Evolution of the Mean Number of Proteins: Initial Point at 1400, equilibrium at 290. Time scale: $\times 10^4$ sec.}\label{FigVersB}
\end{figure}

\providecommand{\bysame}{\leavevmode\hbox to3em{\hrulefill}\thinspace}
\providecommand{\MR}{\relax\ifhmode\unskip\space\fi MR }
\providecommand{\MRhref}[2]{%
  \href{http://www.ams.org/mathscinet-getitem?mr=#1}{#2}
}
\providecommand{\href}[2]{#2}

\appendix
\section{Convergence Results}\label{Appendix}
We first introduce some notations that will be used throughout this section. 
\subsection{Evolution Equations}
We will use  the Skorohod's topology for  convergence in distribution in the space ${\cal D}([0,T],\R_+)$ of  c\`adl\`ag processes.  See Chapter~3 of Billingsley~\cite{Billingsley} for example. To simplify the presentation, all our processes will be defined on the same probability space in the following way. 

Let  ${\cal N}_i^{+}$, ${\cal N}_i^{-}$,  $i{=}1$, $2$, $3$ be independent Poisson processes on $\R_+^2$ with rate $1$ defined on a probability space $(\Omega, {\cal F}, \P)$. 
If  $A\in{\cal B}(\R_+^2)$ is a Borelian subset of $\R_+^2$ and $(i,c)\in\{1,2,3\}\times\{+,-\}$, ${\cal N}_i^{c}(A)$ denotes the number of points of the process ${\cal N}_i^{c}$ in the subset $A$.  
For $t\geq 0$, one denotes by ${\cal F}_t$ the $\sigma$-field generated by the random variables
\[
{\cal N}_i^{c}(B\times[0,t]) \text{ for } B\in{\cal B}(\R_+) \text{ and } (i,c)\in\{1,2,3\}{\times}\{{+},{-}\}.
\]
It is easily seen that the process $(X^N_F(t))$ has the same distribution as the solution of the following stochastic differential equations (SDE)
\begin{align}
\diff I^N_F(t) &= \ind{I^N_F(t-)=0}{\cal N}_1^+([0,\lambda_1^+N]\times[\diff t])\\& \hspace{2cm}-\ind{I^N_F(t-)=1}{\cal N}_1^-([0,\lambda_1^- N P^N_F(t-)]\times[\diff t])\notag\\
\diff M^N_F(t) &= \ind{I^N_F(t-)=1}{\cal N}_2^+([0,\lambda_2 N]\times[\diff t])-{\cal N}_2^-([0,\mu_2N M^N_F(t-)]\times[\diff t])\label{o1}\\
\diff P^N_F(t) &= {\cal N}_3^+([0,\lambda_3 M^N_F(t-)]\times[\diff t])-{\cal N}_3^-([0,\mu_3 P^N_F(t-)]\times[\diff t])\label{o2}
\end{align}
with the same initial condition. For any $N\geq 1$, $(X^N_F(t))$ is a Markov process adapted to the filtration $({\cal F}_t)$. These SDE can be rewritten as,  for some function $f$ with finite support  on ${\cal S}$, 
\begin{multline}\label{SDE}
f(X^N_F(t))=f(X^N_F(0))+
\int_0^t \lambda_1^+N(1-I^N_F(u)) \Delta_1(f)(X^N_F(u))\,\diff u\\+\int_0^t \lambda_1^-NP^N_F(u) I^N_F(u) \Delta_1(f)(X^N_F(u))\,\diff u\\
+\int_0^t \lambda_2N I^N_F(u)\Delta_2^+(f)(X^N_F(u))\,\diff u+\int_0^t \mu_2 N M^N_F(u) \Delta_2^-(f)(X^N_F(u))\,\diff u\\
+\int_0^t \lambda_3 M^N_F(u) \Delta_3^+(f)(X^N_F(u))\,\diff u+\int_0^t \mu_3 P^N_F(u) \Delta_3^-(f)(X^N_F(u))\,\diff u+W^N_f(t)
\end{multline}
where, for $x=(i,m,p)\in{\cal S}$,  the operators $\Delta_{\cdot}^{+/-}$ are defined by
\[
\begin{cases}
\Delta_1(f)(x)=f(1{-}i,m,p){-}f(x)\\
\Delta_2^+(f)(x)=f(i,m+1,p){-}f(x),\quad \Delta_2^{-}(f)(x)=f(i,m{-}1,p){-}f(x)\\
\Delta_3^+(f)(x)=f(i,m,p+1){-}f(x),\quad \Delta_3^{-}(f)(x)=f(i,m,p{-}1){-}f(x),
\end{cases}
\]
and $(W^N_f(t))$ is a local martingale whose previsible increasing process is given by 
\begin{multline}\label{crocW}
\croc{W^N_f}(t){=}
\int_0^t \left[\lambda_1^+N(1{-}I^N_F(u)){+}\lambda_1^-NP^N_F(u) I^N_F(u)\right] \left[\Delta_1(f)(X^N_F(u))\right]^2\,\diff u\\
+\int_0^t \lambda_2N I^N_F(u)\left[\Delta_2^+(f)(X^N_F(u))\right]^2\,\diff u+\int_0^t \mu_2N M^N_F(u) \left[\Delta_2^-(f)(X^N_F(u))\right]^2\,\diff u\\
+\int_0^t \lambda_3 M^N_F(u) \left[\Delta_3^+(f)(X^N_F(u))\right]^2\,\diff u+\int_0^t \mu_3 P^N_F(u) \left[\Delta_3^-(f)(X^N_F(u))\right]^2\,\diff u.
\end{multline}
See Rogers and Williams~\cite{Rogers} for example. 

\begin{definition}\label{def1}
Let   $(\overline{M}^N(t),\overline{P}^N(t))$ be the Markov process with transition rates given by 
\begin{equation}\label{eqr1}
\begin{cases}
(m,p)\rightarrow (m+1,p) \text{ at rate } \lambda_2N, &(m,p)\rightarrow (m-1,p) \text{\phantom{are}}"\text{\phantom{are}}  \mu_2mN,\\
(m,p)\rightarrow (m,p+1) \text{\phantom{are}}"\text{\phantom{are}} \lambda_3 m, & (m,p)\rightarrow (m,p-1)\text{\phantom{are}}"\text{\phantom{are}} \mu_3p
\end{cases}
\end{equation}
and initial state $(\overline{M}^N(0),\overline{P}^N(0))=(m_0,p_0)$. 
\end{definition}
The process $(\overline{M}^N(t),\overline{P}^N(t))$ is simply the analogue of our process $({M}^N_F(t),{P}^N_F(t))$ when the gene is always active. 
\begin{lemma}\label{lem1}
\begin{enumerate}
\item  For the convergence in distribution for the uniform norm on compact sets
  \[
  \lim_{N\to+\infty} \left(\int_0^t \overline{M}^N(u)\,\diff u\right)=(\rho_2 t).
  \]
\item  For $T>0$,
  \[
  \sup_{N\geq 1}\E\left(\sup_{0\leq t\leq T} \overline{P}^N(t)\right)<+\infty.
  \]
\end{enumerate}
\end{lemma}

\begin{proof}
  From Relations~\eqref{eqr1},  it is easily seen that  the process $(\overline{M}^N(t))$ can be expressed $(L_1(Nt))$ where $(L_1(t))$ is an $M/M/\infty$ queue with arrival rate $\lambda_2$ and service rate $\mu_2$ with $L_1(0)=m_0$. See Chapter~6 of Robert~\cite{Robert} for example. Elementary stochastic calculus gives, for $t>0$,
  \begin{equation}\label{eqr2}
  L_1(Nt)=m_0+\lambda_2 N t-\mu_2\int_0^{Nt} L_1(u)\,\diff u +{\cal M}_1^N(t),
  \end{equation}
  where $({\cal M}_1^N(t))$ is a local martingale whose previsible increasing process is given by
  \[
\croc{{\cal M}_1^N}(t)=\lambda_2  N t+\mu_2  \int_0^{Nt}L_1(u)\,\diff u.
\]
Doobs'Inequality shows that the process $({\cal M}_1^N(t)/N)$ vanishes for the convergence in distribution as $N$ gets large. 

For $\eps>0$ and $x\in \N$, if
\[
T_{x}=\inf\{t\geq 0: L_1(u)\geq x\},
\]
Proposition~6.10 of Robert~\cite{Robert} shows the convergence in distribution
\[
\lim_{x\to+\infty} \frac{\rho_2^{x}}{(x-1)!} T_x=E_0
\]
where $E_0$ is an exponential random variable with parameter ${\mu_2\exp(-\rho_2)}$. This shows in particular the process $(L_1(Nt)/N)$ converges in distribution to $0$  for the uniform convergence on compact intervals since
\[
\P\left(\sum_{0\leq t\leq T} \frac{L_1(Nt)}{N}\geq \eps\right)\leq \P\left(T_{\lfloor \eps N\rfloor}\leq N T\right). 
\]
From Equation~\eqref{eqr2}, one gets
  \[
  \int_0^t \overline{M}^N(u)\,\diff u=\frac{1}{N}\int_0^{Nt} L_1(u)\,\diff u =\rho_2t + \frac{1}{\mu_2}\left(\frac{m_0}{N} - \frac{L_1(Nt)}{N} + \frac{{\cal M}_1^N(t)}{N}\right)
  \]
  and therefore assertion~1) of the lemma.

For the last assertion, the method is similar: one first write  the evolution equation
  \[
  \overline{P}^N(t)=  p_0+\lambda_3\int_0^t  \overline{M}^N(u)\,\diff u   -\mu_3\int_0^t\overline{P}^N(u)\,\diff u +{\cal M}_2^N(t),
  \]
 where $({\cal M}_2^N(t))$ is a local martingale whose previsible increasing process is given by
  \[
\croc{{\cal M}_2^N}(t)=\lambda_3\int_0^t  \overline{M}^N(u)\,\diff u +\mu_3  \int_0^{t}  \overline{P}^N(u)\,\diff u.
\]
Define  $\overline{P}^{N}_*(t){=}\sup\{\overline{P}^N(u): 0\leq u\leq t\}$, then for $0\leq t\leq T$
  \begin{multline}\label{eqr4}
\E\left( \overline{P}^N_*(t)\right)\leq   p_0+\lambda_3\E\left(\int_0^T  \overline{M}^N(u)\,\diff u\right) \\+\E\left(\sup_{0\leq u\leq t}|{\cal M}_2^N(u)|\right) +\mu_3\int_0^t\E\left(\overline{P}^N_*(u)\right)\,\diff u.
  \end{multline}
Doob's Inequality gives, for $t\leq T$,
\[
\E\left(\sup_{0\leq u\leq t}|{\cal M}_2^N(u)|\right)\leq 2\lambda_3\int_0^T\E\left( \overline{M}^N(u)\right)\,\diff u +2\mu_3  \int_0^{t} \E\left(\overline{P}^N_*(u)\right)\,\diff u,
\]
and from the ergodic theorem for $(L_1(t))$ (recall that $\overline{M}^N(t)=L_1(Nt)$) one gets
\[
\lim_{N\to+\infty} \E\left(\int_0^T  \overline{M}^N(u)\,\diff u\right)=\rho_3T.
\]
One concludes by using Equation~\eqref{eqr4} and Gronwall's Lemma. 
\end{proof}

\begin{proposition}\label{lem2}
The sequence $(P^N_F(t))$ is tight for the convergence in distribution of c\`adl\`ag processes.
\end{proposition}
\begin{proof}
  Aldous' criterion for tightness is used. See Theorem~4.5 page~320   of Jacod and Shiryaev~\cite{Jacod} for example.  For $T>0$, one denotes by ${\cal T}_T$ the set of stopping times associated to the filtration $({\cal F}_t)$ which are bounded by $T$. For $\eta>0$, let $\tau_1$, $\tau_2\in{\cal T}_T$ be such that $\tau_1\leq \tau_2\leq \tau_1+\eta$.  The respective probabilities that, on the time interval $[\tau_1,\tau_2]$,  no protein is made or that no protein is degraded are respectively  given by
  \[
  \E\left(\exp\left(-\lambda_3\int_{\tau_1}^{\tau_2} M^N_F(u)\,\diff u\right)\right)  \text{ and }
  \E\left(\exp\left(-\mu_3\int_{\tau_1}^{\tau_2} P^N_F(u)\,\diff u\right)\right)  
  \]
By using the strong Markov property, one gets the relation 
\begin{multline*}
\P\left(|P^N_F(\tau_1)-P^N_F(\tau_2)|\geq 1 \right)\leq 1-\E\left(\exp\left(-\lambda_3\int_{\tau_1}^{\tau_2} M^N_F(u)\,\diff u\right)\right)\\+1-\E\left(\exp\left(-\mu_3 \int_{\tau_1}^{\tau_2}P^N_F(u)\,\diff u\right)\right).
\end{multline*}
With a simple coupling using the same Poisson processes ${\cal N}^{+/-}_{2/3}$ of Equations~\eqref{o1} and~\eqref{o2}   gives a process as in  Definition~\ref{def1}  on the same probability space  such that the relations $M^N_F(t)\leq \overline{M}^N(t)$ and $P^N_F(t)\leq \overline{P}^N(t)$ hold almost surely for all $t\geq0$.  From the last relation, one gets the inequality 
\begin{align*}
  \P\left(|P^N_F(\tau_1)-P^N_F(\tau_2)|\geq 1 \right)&\leq 1-\E\left(\exp\left(-\lambda_3\int_{\tau_1}^{\tau_1+\eta} \overline{M}^N(u)\,\diff u\right)\right)\\&\hspace{30mm} +1-\E\left(\exp\left(-\mu_3 \eta\sup_{0\leq t\leq T} \overline{P}^N(t)\right)\right)\\
&\leq 1-\E\left(\exp\left(-\lambda_3 \sup_{0\leq t\leq T}\int_{t}^{t+\eta} \overline{M}^N(u)\,\diff u\right)\right)\\&\hspace{30mm}\hfill +1-\E\left(\exp\left(-\mu_3 \eta\sup_{0\leq t\leq T} \overline{P}^N(t)\right)\right).
\end{align*}
Lemma~\ref{lem1} gives the relation
\[
\lim_{N\to+\infty} \sup_{\tau_1\in{\cal T}_T} \E\left(\exp\left(-\lambda_3\int_{\tau_1}^{\tau_1+\eta} \overline{M}^N(u)\,\diff u\right)\right)= e^{-\lambda_3\rho_2\eta}
\]
and, for $\eps>0$, the existence of $K>0$ such that
\[
\sup_{N\geq 1} \P\left(\sup_{0\leq t\leq T} \overline{P}^N(t) \geq K \right)\leq \eps.
\]
Consequently
\[
\lim_{\eta\to 0}\lim_{N\to+\infty} \sup_{\substack{\tau_1,\tau_2\in{\cal T}_T\\\tau_1\leq \leq\tau_2\leq \tau_1+\eta}}   \P\left(|P^N_F(\tau_1)-P^N_F(\tau_2)|\geq 1 \right)=0,
\]
hence,  by Aldous' criterion, the tightness of the sequence $(P^N_F(t))$ is established. The proposition is proved.

\end{proof}
\subsection{Convergence of Occupation Measures}\label{OccMeas}
For $N\geq 1$ and $T>0$, one defines the random measure $\Lambda^N$ on ${\cal E}_T\stackrel{\text{def.}}{=}\{0,1\}\times\N^2\times[0,T]$  as follows,  for a non-negative Borelian function $G$ on ${\cal E}_T$, 
\[
\Lambda^N(G)=\int_0^{T} G(X^N_F(u),u)\,\diff u.
\]
If $A$ is a Borelian subset of ${\cal E}_T$,  $\Lambda^N(A)$ denotes $\Lambda^N(\mathbbm{1}_A)$. 
\begin{proposition}\label{occmes}
The sequence $\Lambda^N$ of random measures is tight and any of its limiting points $\Lambda$ can be written as
\[
\Lambda(F)=\sum_{(i,m,p)\in{\cal S}}\int_0^{T} G(i,m,p,u)\pi_{p}(i,m)\nu_u(p)\,\diff u.
\]
where, for any $u\leq T$, $\nu_u$ is a positive measure on $\N$ such that, almost surely,
\[
\int_0^t \nu_u(\N)\,\diff u=t,\quad \forall t\leq T,
\]
and, for $p\in\N$, $\pi_{p}$ is the invariant distribution of the Markov process on $\{0,1\}\times \N$ whose transition rates are given by, for $(i,m)\in\{0,1\}\times \N$,
\begin{equation}\label{MRapide}
\begin{cases}
(i,m)\rightarrow (1-i,m) \quad\text{ at rate } \lambda_1^+ i+\lambda_1^-p(1-i),\\
(i,m)\rightarrow (i,m+1) \text{\phantom{at rate}} \lambda_2 i,\\
 (i,m)\rightarrow (i,m-1) \text{\phantom{at rate}} \mu_2m. 
\end{cases}
\end{equation}
Additionally, one has 
\begin{equation}\label{eqEM}
\sum_{(i,m)\in\{0,1\}\times\N} m\pi_p(i,m)=\frac{\lambda_1^+}{\lambda_1^++\lambda_1^-p}\frac{\lambda_2}{\mu_2}.
\end{equation}
\end{proposition}
\begin{proof}
For $K{>}0$, if ${\cal K}_K$ is the compact subset $\{0,1\}{\times}[0,K]^2{\times}[0,T]$ of ${\cal E}_T$, then
\[
\E\left(\Lambda^N\left({\cal E}_T{\setminus}{\cal K}_K\right)\right) \leq \int_0^T\P\left(M_F^N(u)\geq K\right) \,\diff u +T\P\left(\sup_{0\leq u\leq T} P_F^N(u)\leq K\right).
\]
By using the same coupling as in the proof of Proposition~\ref{lem2}, one gets that
\[
\E\left(\Lambda^N\left({\cal E}_T{\setminus}{\cal K}_K\right)\right) \leq \int_0^T\P\left(\overline{M}^N(u)\geq K\right) \,\diff u +T\P\left(\sup_{0\leq u\leq T} \overline{P}^N(u)\leq K\right).
\]
By Lemma~\ref{lem1}, for $\eps>0$, there exists some $K$ such that
\[
\sup_{N\geq 1} \E\left(\Lambda^N\left({\cal E}_T{\setminus}{\cal K}_K\right)\right) \leq \eps. 
\]
 Consequently, the sequence $(\Lambda^N)$ of random Radon measures on ${\cal E}_T$ is tight. See Dawson~\cite[Lemma~3.28, page~44]{Dawson} for example.

Let $\Lambda$ be a limiting point of some subsequence $(\Lambda^{N_k}(\cdot))$. By using Radon-Nikodym's Theorem,  see Chapter~8 of Rudin~\cite{Rudin:01} for example,  it is not difficult to see that there exists some non-negative  random variables $(\ell_u(x)(\omega), (\omega,x,u)\in\Omega\times{\cal S}\times[0,T])$ such that
$(\omega,x,u)\mapsto \ell_u(x)(\omega)$ is measurable and $\Lambda$ can be expressed as
\[
\Lambda(G)=\sum_{x\in{\cal S}}\int_0^T G(x,u)\ell_u(x)\,\diff u.
\]
From the domination relation of Lemma~\ref{lem1}, one gets that, almost surely, there is no loss of mass, i.e.
\begin{equation}\label{mass}
\int_0^t \ell_u({\cal S})\,\diff u=t, \quad \forall t\leq T,
\end{equation}
holds almost surely. 
Now take  a function $f$ with bounded support on ${\cal S}$, by using Equation~\eqref{crocW}, it is not difficult to show that the process $(\langle{W_f^N}\rangle(t))$ satisfies the relation
\[
\lim_{N\to+\infty} \frac{1}{N^2}\E\left(\croc{W^N_f}(T)\right)=0,
\]
by Doob's Inequality this implies that the martingale $(W_f^N(t)/N)$ converges in distribution to $0$ for the uniform norm on $[0,T]$. 

By dividing Relation~\eqref{SDE} by $N$, one gets that, for the convergence in distribution, the relation 
\begin{multline*}
\lim_{N\to+\infty} \left(\int_0^t \lambda_1^+\left(1{-}I^N_F(u)\right) \Delta_1(f)(X^N_F(u))\,\diff u\right.\\{+}\int_0^t \lambda_1P^N_F(u) I^N_F(u)\Delta_1(f)(X^N_F(u))\,\diff u\\\left.
{+}\int_0^t \lambda_2 I^N_F(u) \Delta_2^+(f)(X^N_F(u))\,\diff u{+}\int_0^t \mu_2 M^N_F(u) \Delta_2^-(f)(X^N_F(u))\,\diff u\right){=}0.
\end{multline*}
holds.  The convergence of the sequence $(\Lambda^{N_k})$ gives that the relation 
\begin{multline*}
\sum_{x=(i,m,p)\in{\cal S}} \int_0^t \lambda_1^+ (1{-}i)\Delta_1(f)(x)\ell_u(x)\,\diff u{+}\int_0^t \lambda_1^- i p \Delta_1(f)(x)\ell_u(x)\,\diff u\\
+\int_0^t \lambda_2 i \Delta_2^+(f)(x)\ell_u(x)\,\diff u+\int_0^t \mu_2 m\Delta_2^-(f)(x)\ell_u(x)\,\diff u=0
\end{multline*}
holds almost surely for all $0\leq t\leq T$ and for all indicator functions of elements ${\cal S}$. 
Now, for $p\in\N$ and $g$  a function with finite support on $\{0,1\}\times\N$,   define $f(i,m,p)=g(i,m)$, the above relation gives
\begin{multline}\label{InvAux}
\sum_{x=(i,m,p)\in{\cal S}}  \ell_u(i,m,p) \lambda_1^+ (1{-}i)\Delta_1(g)(i,m){+} \ell_u(i,m,p)\lambda_1^- i p\Delta_1(g)(i,m)\\
+ \ell_u(i,m,p)\lambda_2 i \Delta_2^+(g)(i,m)+ \ell_u(i,m,p)\mu_2 m\Delta_2^-(g)(i,m)=0
\end{multline}
holds almost surely for all $u\in{\cal A}\subset [0,T]$ and $[0,T]{-}{\cal A}$ is  negligible for Lebesgue measure.
Relation~\eqref{InvAux} shows that for $u\in{\cal A}$, the vector $(\ell_u(i,m,p))$ is proportional to the invariant distribution $\pi_{p}$ of the Markov process on $\{0,1\}\times \N$ whose transition rates are given by Relations~\eqref{MRapide}.

One gets therefore the existence of a constant $\nu_u(p)$ such that $\ell_u(i,m,p)=\nu_u(p)\pi_p(i,m)$ for all $(i,m,p)\in{\cal S}$.  Equation~\eqref{mass} gives the relation
\[
\int_0^t \nu_u(\N)\,\diff u=t, \quad \forall t\leq T.
\]
Hence one has $\nu_u(\N)=1$  almost surely for all $u\in{\cal A}_1\subset [0,T]$ and $[0,T]{-}{\cal A}_1$ is  negligible for Lebesgue measure.

Straightforward calculations as in the proof of Proposition~\ref{FeedProp} complete the proof of the proposition to give Relation~\eqref{eqEM}.
\end{proof}

\end{document}